\documentclass{amsart}
\usepackage{eurosym}
\usepackage{amsfonts}
\usepackage{amssymb}
\usepackage{epsfig}
\usepackage{verbatim}
\usepackage{graphicx}
\usepackage{color}
\usepackage{url}
\usepackage{multicol}

\def\cal{\mathcal}

\def\cali{{\cal I}}

\let\phi\varphi

\definecolor{Comment}{rgb}{0.5,0,0.5}

\let\tilde\widetilde

\newtheorem{theorem}{Theorem}
\newtheorem{corollary}{Corollary}
\newtheorem{lemma}{Lemma}

\begin{document}
\title[Probability of infection in an SIR network is NP-Hard]{Finding the
probability of infection in an SIR network is NP-Hard}
\author{Michael Shapiro}
\address{Department of Pathology, Tufts University, Boston, MA 02111}
\email{michael.shapiro@tufts.edu}
\thanks{The first author wishes to acknowledge support from NIH Grant K25
AI079404-02.}
\author{Edgar Delgado-Eckert}
\address{Department of Biosystems Science and Engineering, Swiss Federal
Institute of Technology Zurich (ETH Z\"urich), Basel, Switzerland and Swiss
Institute of Bioinformatics.}

\begin{abstract}
{A common approach in epidemiology is to study the transmission of a disease
in a population where each individual is initially susceptible (S), may
become infective (I) and then removed or recovered (R) and plays no further
epidemiological role. Much of the recent work gives explicit consideration
to the network of social interactions or disease-transmitting contacts and
attendant probability of transmission for each interacting pair. The state
of such a network is an assignment of the values ${\{S,I,R\}}$ to its
members. Given such a network, an initial state and a particular susceptible
individual, we would like to compute their probability of becoming infected
in the course of an epidemic. It turns out that this and related problems
are \textbf{NP}-hard. In particular, it belongs in a class of problems for
which no efficient algorithms for their solution are known. Moreover,
finding an efficient algorithm for the solution of any problem in this class
would entail a major breakthrough in theoretical computer science.}
\end{abstract}
\maketitle



\section{Introduction}

\label{sec:introduction} Mathematical modelling of epidemics is often traced
to the {celebrated} SIR model of Kermack and McKendrick \cite{KM}. This
model posits a population of constant size whose members fall into one of
three classes: susceptible (S), infective (I) and removed (R). Approximating
these as continuous and assuming well-mixing, i.e., each individual is in
equal contact with and equally likely to infect each other individual,
allows for an approximate description of the infection dynamics using
ordinary differential equations ({ODE)}.

Clearly, as it has been argued by many in theoretical \cite%
{Bansal22102007,N1,Eubank2004} as well as experimental studies \cite%
{Cauchemez15022011}, the well-mixing assumption is not an accurate
representation of real contact patterns. Thus, much recent work has focused
on the role of the network of disease-transmitting contacts. (Reviewed in 
\cite{Keeling22092005}. See also, \cite%
{Newman03thestructure,eubankNetworkBased}. For a comparison of well-mixed
and network-based models, see \cite{Keeling20051}.) Indeed, Kermack's and
McKendrick's{\ ODE model arises as the limiting case of a simplistic network
model in which each individual has an equal chance of infecting every other.}
However, real-world social contact networks exhibit complex patterns of
interconnection between individuals. Further, the probability of
transmitting disease from one individual to another depends on the nature,
frequency and duration of the contact as well as the immune competence of
the target individual. This leads to a modelling formalism of \emph{social
networks} as a probabilistic graph $\mathcal{G}=(G,{\mathop{\text{Pr}}})$.
Here $G$ is the graph $G=(V,E)$, each vertex $u\in V$ is an individual, each
edge $e=(u,v)\in E$ records the fact that $u$ might infect $v$ and ${%
\mathop{\text{Pr}}}:E\rightarrow \lbrack 0,1]$ gives the probability that $u$
infects $v$ if $u$ becomes infective while $v$ is susceptible. In this
formalism, $\mathcal{G}$ is a fixed graph $G$ with labelling ${%
\mathop{\text{Pr}}}$.

This relatively new modelling paradigm has triggered an enormous amount of
research in theoretical epidemiology. The field has greatly benefited from
approaches that range from applications of bond percolation theory and other
techniques from statistical physics \cite{PhysRevE.79.036113,
  PhysRevE.80.020901, N1, Eur.Phys.J.B_26, Meyers2006400,
  SIRinRandomNetworks,Volz07122007, PhysRevLett.86.3200, House06012011,
  PhysRevLett.90.028701} to large scale simulation endeavours \cite%
{Eubank2004, stroud2007, Mniszewski:2008:ESM:1400549.1400636,
  EpiSimSLA,Mniszewski:2008:PSA:1394977.1394987}. Given that this mathematical
formalism seems accurate and powerful to describe the spread of infectious
diseases, the natural question arises as to whether calculations performed
within this formalism can be used in practical situations to make useful
predictions. Such calculations are based on potentially measurable parameters
such as network topology and transmission probabilities \cite%
{Wallinga15_2006}. For instance, one could attempt to calculate the
probability that, given a social contact network $\mathcal{G}$, an epidemic
starting with a set $P$ of infectives results in the infection of an initially
susceptible individual $u$. Are there any computational limitations when
trying to calculate such magnitudes? If yes, how limiting are they?
Fortunately, to address the computational issues associated with this and
similar calculations, we don't need to start from scratch, given that network
engineers have already studied since the 1970s problems that are essentially
the same.

In the era of electronically digitalized information and digital computers,
communications networks have become the biggest and count among the most
important networks. The size of these networks is exponentially increasing.
For instance, the size of the Internet shows exponential growth since its
creation in the early nineties (http://www.isc.org/). As the components of
such networks are subject to failure, engineers face the problem of
designing, constructing and operating networks that meet the required
standards of reliability. Of particular interest is the estimation of how
reliable a given network is in performing its function, provided some
knowledge about the reliability of its components is available. In many
cases, the functionality of the network can be expressed as the ability of
its topology to support the network's operation. In other words, the network
is functional if and only if certain connectivity properties are fulfilled.
Consider a network of computers which use this network to transmit messages.
Let us suppose that each of these computers is reliable, but that each
communication link has some chance of failure when called upon to transmit a
message. We then encounter the same formalism explained above for social
networks. A\emph{\ communications network} is given by $\mathcal{G}=(G,{%
\mathop{\text{Pr}}})$ where each vertex $u\in V$ is a computer, each edge $%
e=(u,v)\in E$ is a communication link and ${\mathop{\text{Pr}}}:E\rightarrow
\lbrack 0,1]$ is the reliability of the communication link from $u$ to $v$.
One might ask, given a communications network $\mathcal{G}$, a set of
computers $P$ and a computer $u\notin P$, if the computers in $P$ all send a
message, what is the chance it will reach $u$? We will see that this is the
same problem we stated above in the context of epidemics on social contact networks.

It has long been known in the communications network literature that this
problem is computationally intractable. A standard benchmark of
computational complexity is the class of \textbf{NP}-complete\ problems.
This class has the following properties:

\begin{itemize}
\item At present, no algorithm for an \textbf{NP}-complete problem is known
to have a running time which is bounded by a polynomial. Indeed, many
algorithms for \textbf{NP}-complete problems have exponential running time.
It is unknown whether any \textbf{NP}-complete problem can be solved by an
algorithm with polynomial running time.

\item If any problem in this class can be solved by an algorithm whose
running time is bounded by a polynomial, then every problem in this class
can be solved by an algorithm whose running time is bounded by a polynomial.
\end{itemize}

In view of the second, it is considered unlikely that any \textbf{NP}
-complete\ problem has a polynomial time solution. The communication among
computers problem (and hence the epidemiology problem) listed above is known
to be as hard as any \textbf{NP}-complete\ problem. Such problems are termed 
\textbf{NP}-hard. This is not the first problem in network epidemiology
known to be \textbf{NP}-hard. Previously known examples include the
following: Given a social contact network and limited resources

\begin{itemize}
\item What is the optimal strategy for vaccinating a limited number of
individuals?

\item What is the optimal strategy for quarantining a limited number of
individuals?

\item What is the optimal strategy for placement of a limited number of
sensors for monitoring the course of an epidemic?
\end{itemize}

(See \cite{wangModeling, eubankStructural, aspnesSODA,
hayrapetyanUnbalancedGraphCuts}.) These problems involve the search for an
optimum among subsets of the vertices or edges of the given social contact network.
It might be hoped that finding the probability of infection of a single
individual would be computationally less demanding. As the engineers have
taught us, this is not so. While this result has been recently reported in
the physics and operations research community \cite{Laumanns_Zenklusen_2009}%
, it seems almost unknown among epidemiologists.

This article is organized as follows: In Section \ref%
{sec:computationalComplexity} we give a very brief overview of the relevant
concepts and methods in computational complexity. This provides the
unacquainted reader with the basic tools for understanding the main message
of this paper. Section \ref{sec:epidemicsOnNetworks} provides the elementary
formal mathematical framework for studying SIR epidemics on networks,
including the connection with percolation theory. In Section \ref%
{sec:communicationNetworks} we present a series of problems that have been
studied in network engineering and demonstrate their structural isomorphism
with certain problems concerning SIR epidemics on networks. Section \ref%
{sec:problemsInEpidemiology} is devoted to studying the computational
complexity of extended/generalized epidemiological problems. We finish in
Section \ref{sec:discussion} with some concluding remarks.


\section{Computational complexity}

\label{sec:computationalComplexity}

In this section we give a brief account of the class \textbf{NP}-complete.
This class is a common benchmark for describing problems which are
algorithmically soluble but computationally intractable. For those wishing a
fuller account we recommend \cite{gareyComputers}.

In describing the class \textbf{NP}-complete, it is useful to describe the
class \textbf{P}, and necessary to describe the class \textbf{NP}. These
classes of problems are defined in terms of computational complexity.

The computational complexity of a problem $\Pi$ is measured in terms of the
running time necessary for an algorithm which solves $\Pi$. Defining these
terms requires some preliminaries. First, note that a \emph{problem} $\Pi$
consists of a collection of \emph{instances}, $D_\Pi$. Thus, ``Determine
whether 18 is composite'' is an instance of the problem, ``For any integer $%
n $, determine whether $n$ is composite.'' This is an example of a \emph{\
decision problem}, that is, for each instance, the answer is either ``yes''
or ``no''. A decision problem $\Pi$ can be formalized as the pair $%
(D_\Pi,Y_\Pi)$, where $Y_\Pi \subset D_\Pi$ consists of the \emph{yes
instances}. In this example, $D_\Pi$ is the set of integers and $Y_\Pi$ is
the set of composite integers. We will refer to this problem as $\Pi_\text{%
composite}$.

Notice that each instance $\pi \in \Pi $ has a \emph{size}, $\ell (\pi )$
and that the computational cost of solving the problem grows with the size
of the problem. In this example, the size $\ell (n)$ of the instance $n$ is
the number of digits in $n$. If we then have an algorithm $M$ which solves $%
\Pi $, we can consider the running time $r_{M}(\pi )$ required by $M$ when
applied to the instance $\pi $. This could be measured in elapsed time or in
terms of the number of steps carried out by $M$ in this computation. We can
then define the \emph{running time} of $M$ to be 
\begin{equation*}
r_{M}(n)=%
\begin{cases}
0 &  \\ 
\hskip.2in\text{if $\{\pi \mid \ell (\pi )=n\}=\emptyset $} &  \\ 
\max \{r_{M}(\pi )\mid \ell (\pi )=n\} &  \\ 
\hskip.2in\text{otherwise} & 
\end{cases}%
\end{equation*}%
The class \textbf{P}\ consists of those decision problems which can be
solved with a polynomial running time. Stated formally, a decision problem $%
\Pi $ belongs to the class \textbf{P}\ if there is an algorithm $M$ which
solves $\Pi $ and a polynomial $p(n)$ such that $r_{M}(n)\leq p(n)$. An
example of a problem in the class \textbf{P}\ is $\Pi _{\text{mult}}$. An
instance of $\Pi _{\text{mult}}$ is three integers, $a$, $b$ and $c$. The
size of an instance is the total number of digits in $a$, $b$ and $c$. These
constitute a yes instance if $a\times b=c$.

The class \textbf{NP}\ consists of non-deterministic polynomial time
problems. That is, a decision problem is \textbf{NP}\ if a machine which is
allowed to guess can verify a yes instance in polynomial time. $\Pi _{\text{%
composite}}$ provides and example of a problem which is \textbf{NP}. Given
an instance of $\Pi _{\text{composite}}$, i.e., an integer $c$, if $c$ is,
in fact, composite, a correct guess as to its factors $a$ and $b$, can be
verified in polynomial time by calling $\Pi _{\text{mult}}$. One can define
this class in terms of the operation of non-deterministic Turing machines.
See, for example, \cite{HopcroftUllman}. Clearly $\mathbf{P}\subseteq 
\mathbf{NP}$. In view of the perceived complexity of many problems in 
\textbf{NP}, it is generally believed that $\mathbf{P}\neq \mathbf{NP}$.

The class \textbf{NP}-complete\ consists of the hardest problems in \textbf{%
NP}. The problems in \textbf{NP}-complete\ have the following property:
Suppose that $\Pi _{1}$ is \textbf{NP}-complete. Suppose that $\Pi _{2}$ is 
\textbf{NP}. Then there is an algorithm $M$ which translates any instance $%
\pi _{2}$ of $\Pi _{2}$ into an instance $\pi _{1}$ of $\Pi _{1}$ such that $%
\pi _{1}$ is a yes instance of $\Pi _{1}$ if and only if $\pi _{2}$ is a yes
instance of $\Pi _{2}$. Further, both the computational cost of translating $%
\pi _{2}$ into $\pi _{1}$ and the size $\ell (\pi _{1})$ are bounded by a
polynomial in $\ell (\pi _{2})$. It follows that if any \textbf{NP}%
-complete\ problem can be solved (deterministically) in polynomial time,
then every \textbf{NP}\ problem can be solved in polynomial time. Put
another way, if any \textbf{NP}-complete\ problem can be solved in
polynomial time, we will then have $\mathbf{P}=\mathbf{NP}$.

Hundreds of problems are known to be \textbf{NP}-complete\ \cite%
{gareyComputers}. These come from fields such as graph theory, number
theory, scheduling, code optimization and many others. They are widely
believed to be intrinsically intractable, but this remains an open question.
Other problems which are not necessarily \textbf{NP}-complete\ (e.g.,
because they are not decision problems) are known to be at least as hard.
This is because for such a problem, say $\Gamma ,$ there is an \textbf{NP}%
-complete problem $\Pi $ that can be \emph{reduced} to $\Gamma ,$ where the
computational cost of this reduction is bounded by a polynomial in the
length of the instance problem considered. Thus, $\Gamma $ can be used to
solve $\Pi .$ These problems are called \textbf{NP}-hard. Since \textbf{NP}%
-complete problems transform to each other, all \textbf{NP}-complete
problems can be solved by a reduction to an \textbf{NP}-hard\ problem. 
\textbf{NP}-hard\ problems are found in fields as diverse as epidemiology
and origami \cite{origami}.


\section{SIR epidemics on networks}

\label{sec:epidemicsOnNetworks}

We start by describing a network SIR model in which both the population and
the individual transmission probabilities are constant with respect to time.

A state of this system is the assignment of each individual to one of the
classes S, I or R. The transmission probabilities determine who can infect
whom and consequently which states can follow a given state. Indeed, they
also determine the probability that any one of these states follows the
given state. An epidemic is a sequence of states each of which is a possible
successor of the previous state. Consequently, given an initial state, we
can speak of the probability that an epidemic evolves through a given
sequence of states and the probability that it arrives at a particular
state. Let us formalize this.

As above, a \emph{social contact network} is a pair $\mathcal{G}=(G,{\ %
\mathop{\text{Pr}}})$ where $G$ is the graph with vertex set $V$ and edge
set $E$. Each edge has the form $(u,v)$ with $u,v\in V$ and $u\neq v$. The
function ${\mathop{\text{Pr}}}$ assigns a probability to each edge, that is $%
{\mathop{\text{Pr}}}:E\rightarrow \lbrack 0,1]$. The \emph{states} of $%
\mathcal{G}$ are given by\footnote{%
In particular, a state $\phi $ can be seen as a subset of the Cartesian
product $V\times $ ${\{S,I,R\},}$ and therefore, it is meaningful to speak
of the probability of a state or of a collection of states.} 
\begin{equation*}
\mathop{\rm St}(\mathcal{G})=\{\phi \mid \phi :V\rightarrow {\{S,I,R\}}\}.
\end{equation*}

Given states $\phi _{1}$ and $\phi _{2}$, the state $\phi _{2}$ is a \emph{%
possible successor} of $\phi _{1}$ if it satisfies the following conditions:

\begin{enumerate}
\item If $\phi_1(u) = R$, then $\phi_2(u) = R$. (Recovered individuals stay
recovered.)

\item If $\phi_1(u) = I$, then $\phi_2(u) = R$. (Infected individuals
recover in one step.)

\item If $\phi _{1}(u)=S$, then $\phi _{2}(u)\in \{S,I\}$. (Susceptible
individuals either stay susceptible or become infected.)

\item If $\phi _{2}(u)=I$, then $\phi _{1}(u)=S$ and there is a vertex $q\in
V\backslash \{u\}$ and an edge $(q,u)$ with $\phi _{1}(q)=I$. (Infected
individuals were susceptible and were infected by a neighbour.)
\end{enumerate}

The requirement that individuals recover in exactly on time-step might
appear to be a drastic oversimplification. However, the formalism is rich
enough to accommodate patterns of latency and extended periods of
infectivity. This can be done by replacing the individual represented by
vertex $u$ by a sequence of vertices $u_{1},u_{2},\dots $ representing $u$
on day 1, $u$ on day 2, etc. See, e.g., \cite{Floyd}.

An \emph{epidemic} $\Phi $ is a sequence of states $\phi _{1},\dots ,\phi
_{k}$ where $\phi _{i+1}$ is a possible successor of $\phi _{i}$ for $%
i=1,\dots ,k-1$. The length of this epidemic is $\ell (\Phi )=k$. Since
individuals recover after one step, infection must be transmitted or die
out. As a consequence, no epidemic can be longer than the longest
self-avoiding path in $G$, for otherwise, it must infect some vertex twice.
If we assume that each edge transmits or fails to transmit independently,
then it is not hard to compute the probability that a susceptible individual
is infected by its infected neighbours. This, in turn, allows one to compute
the probability that a state $\phi _{1}$ is followed by a particular
successor state $\phi _{2}$. Let us denote this probability by ${%
\mathop{\text{Pr}}}(\phi _{2}\mid \phi _{1})$. This system enjoys the Markov
property, that is, the probability of a given state depends only on the
previous state. Thus given an initial state $\phi _{1}$, the probability of
the epidemic $\Phi =\phi _{1},\dots ,\phi _{n}$, is 
\begin{equation*}
{\mathop{\text{Pr}}}(\Phi \mid \phi _{1})=\prod_{i=2}^{n}{\mathop{\text{Pr}}}%
(\phi _{i}\mid \phi _{i-1}).
\end{equation*}%
The probability that $u$ becomes infected at the $n^{\text{th}}$ step in the
course of an epidemic starting with $\phi _{1}$ is

\begin{equation*}
{\mathop{\text{Pr}}} (\phi _{n}(u)=I\mid \phi _{1})= 
 \sum_{\{\phi _{1},\dots
,\phi _{n}\mid \atop \phi _{n}(u)=I\}} {\mathop{\text{Pr}}}(\phi _{1},\dots ,\phi
_{n}\mid \phi _{1}).
\end{equation*}%
Abusing notation, we denote the probability that $u$ becomes infected in the
course of some epidemic starting with $\phi _{1}$ by 
\begin{equation*}
{\mathop{\text{Pr}}}(u\mid \phi _{1})= 
\sum_{j=1}^{n}{\mathop{\text{Pr}}}(\phi _{j}(u)=I\mid \phi _{1}).
\end{equation*}

Note that since an infected individual becomes recovered at the next stage,
no epidemic appearing in this sum is an initial sub-epidemic of another.
Accordingly, these are disjoint cases.

We will be interested in initial states $\phi_1$ consisting only of
infectives and susceptibles. In this case, we can identify $\phi_1$ with the
set of infectives $P = \phi_1^{-1}(I)$. This gives the notation ${\ %
\mathop{\text{Pr}}}(u \mid P)$.

Let us formalize the problem $\Pi_{\text{epidemic}}$ of finding ${\ %
\mathop{\text{Pr}}}(u \mid P)$. An instance $\pi$ of this problem consists of

\begin{itemize}
\item A graph $G = (V,E)$.

\item A labelling%
\begin{footnote}{There are technical issues here concerning
  the values of these probabilities. To avoid these issues they are
  usually assumed to be rational numbers and bounds are placed on
  the sizes of their denominators.  For details, see
  \cite{colbournCombinatorics}. Since $\mathbb{Q}$ is dense in
  $\mathbb{R}$, this is not a limitation on the possible probability
  values relevant in real applications.}\end{footnote}  
${\mathop{\text{Pr}}}:E\rightarrow \lbrack 0,1]\cap 
\mathbb{Q}
$.

\item An initial infective set $P \subset V$.

\item An individual $u\in V \setminus P$.
\end{itemize}

A solution to $\pi$ is the value ${\mathop{\text{Pr}}}(u \mid P)$.

We take $\ell (\Pi )=\left\vert V\right\vert $.

The epidemiological viewpoint we have just described follows the evolution
of probabilities over time. If we ignore the order of events, we come to the
simpler viewpoint of \emph{percolation}. Percolation methods have been used
in epidemiology. (See, for example, \cite{G,S,Floyd,N1,secondLook,M}. The
latter two contain extensive references.) Since an individual is only
infected for one time step in the course of any epidemic, an edge can
transmit at most once in the course of an epidemic. This allows us to
consider a random variable that takes as values subgraphs of $G$. Given $%
\mathcal{G}$, we take ${\mathbb{G}}$ to be the random variable which takes
values in $\{G^{\prime }=(V,E^{\prime })\mid E^{\prime }\subseteq E\}$. The
probability that ${\mathbb{G}}$ takes the value $G^{\prime}$ is given by 
\begin{equation*}
{\mathop{\text{Pr}}}(G^{\prime })= 
\left( \prod_{e\in E^{\prime }}{\ \mathop{\text{Pr}}}(e)\right) \left(
\prod_{e\notin E^{\prime }}(1-{\ \mathop{\text{Pr}}}(e))\right) .
\end{equation*}
We may think of $E^{\prime }$ as determining whether $e=(u,v)$ transmits in
the course of an epidemic if that epidemic has a state $\phi $ with $\phi
(u)=I$ and $\phi (q)=S$. Given a path $\tau $ in $G$, we will abuse notation
by writing $\tau \subset G$ and $e\in \tau $ for the edges of $\tau $. Given
a path $\tau $, the probability that it appears in $G^{\prime }=(V,E^{\prime
})$ is 
\begin{equation}
{\mathop{\text{Pr}}}\left( \{G^{\prime }\mid \tau \subset G^{\prime
}\}\right) =\prod_{e\in \tau }{\mathop{\text{Pr}}}(e).
\label{eq:probSubgraph}
\end{equation}
For a proof of the following theorem, see, e.g., \cite{Floyd}.

\begin{theorem}
\label{thm:percolation} Suppose $\mathcal{G}$ is social contact network. Then 
\begin{equation*}
{\mathop{\text{Pr}}}(u\mid P)= 
{\mathop{\text{Pr}}}\left( \{G^{\prime }\mid \text{$G^{\prime }$ contains
a path from $P$ to $u$}\}\right) .
\end{equation*}%
In particular, ${\mathop{\text{Pr}}}(u\mid P)$ is a finite sum of terms of
the form (\ref{eq:probSubgraph}). Accordingly, it is a polynomial in the
values ${\mathop{\text{Pr}}}(e)$ with integer coefficients and degree $%
\left\vert E\right\vert $ .
\end{theorem}

This theorem provides the link between epidemiology and communications
networks.


\section{\textbf{NP}-hard\ problems on communications networks: Consequences
for epidemiological calculations}

\label{sec:communicationNetworks}

We assume that a communications network consists of a set of computers, each
of which is reliable and a set of communication links each of which has a
known likelihood of failure and that the communication links function or
fail independently. There is no loss of generality in regarding each node as
infallible, since a fallible computer can be modelled as a pair of nodes
with a fallible link connecting its input to its output. Once again, we can
formalize this as $\mathcal{G}=(G,{\mathop{\text{Pr}}})$, where $G(V,E)$
represents installed capacity ($V$ being the set of computers and $E$ the
set of communication links), ${\mathop{\text{Pr}}}:E\rightarrow \lbrack 0,1]$
the reliability of each link and ${\mathbb{G}}$ is the random variable
assuming values in $\{G^{\prime }=(V,E^{\prime })\mid E^{\prime }\subseteq
E\}.$ Each $G^{\prime }=(V,E^{\prime })$ is the subnetwork of functioning
links left after the failure of the edges $e\in E\backslash E^{\prime }$.
Successful transmission of a message on this network depends on the
connectivity of the subgraph realized by ${\mathbb{G}}$. Network engineers
focus on several kinds of connectivity. We first examine two of the simplest.

The \emph{two-terminal reliability problem} is defined as the calculation of
the probability that there is at least one correctly functioning path in the
network connecting a predefined source node to a predefined target node. An
instance $\pi $ of $\Pi _{\text{ two terminal}}$ consists of the following:

\begin{itemize}
\item A graph $G = (V,E)$.

\item A labelling ${\mathop{\text{Pr}}}:E\rightarrow \lbrack 0,1]\cap 
\mathbb{Q}
$.

\item A source terminal $u\subset V$.

\item A target terminal $v\in V\setminus \{u\}$.
\end{itemize}

A solution to $\pi $ is the value ${\mathop{\text{Pr}}}(v\mid u)$.

By Theorem~\ref{thm:percolation}, this value is an integer polynomial in the
values ${\mathop{\text{Pr}}}(E)$. Thus, if we restrict to the case where ${\ %
\mathop{\text{Pr}}}(E)$ takes a single value, this becomes an integer
polynomial in one variable called the \emph{reliability polynomial}. Thus a
related problem is the following:

An instance $\pi$ of $\Pi_{\text{rel poly}}$ is

\begin{itemize}
\item A graph $G = (V,E)$.

\item A source terminal $u \subset V$.

\item A target terminal $v\in V \setminus \{u\}$.
\end{itemize}

A solution to $\pi $ is the coefficients of the reliability polynomial. A
number of additional network reliability problems have been studied (see 
\cite{colbournCombinatorics}, an excellent introduction to this field).
These include

\begin{itemize}
\item $k$ terminal reliability. This requires that $k$ chosen terminals are
mutually pair wise connected.

\item Broadcasting, also known as all terminal reliability: This requires
that all terminals are pair wise connected.
\end{itemize}

Naturally, in addition to the network reliability problems presented above,
many other reasonable problems can be defined or could arise from practical
applications. Formally, once a model $\mathcal{G}=(G,{\mathop{\text{Pr}}})$
of the network has been chosen, a general mechanism to define a reliability
problem is the following: A network operation is specified by defining a set 
$Op(G)\subseteq \{G^{\prime }=(V,E^{\prime })\mid E^{\prime }\subseteq E\}$
of states considered to be functional. The set $Op(G)$ is sometimes called \
a \emph{stochastic binary system}; the elements of $Op(G)$ are termed \emph{%
\ pathsets}. Specifying the pathsets for $G$ determines the whole stochastic
binary system, and therefore defines the network operation. The reliability
problem consists of finding the probability $\Pr (Op(G))$ that the
probabilistic graph ${\mathbb{G}}$ assumes values in the set $Op(G).$

A first naive algorithm to solve a network reliability problem formulated in
this general manner is to enumerate all states of ${\mathbb{G}}$ (i.e., the
cardinality of the set $\{G^{\prime }=(V,E^{\prime })\mid E^{\prime
}\subseteq E\}$), determine whether a given state is a pathset or not using
some predesigned recognition procedure\footnote{%
Such recognition procedures generally boil down to path-finding or spanning
tree methods, which are efficient (i.e., of polynomial running time) and
well-know procedures in algorithmic graph theory and computer science.}, and
sum the occurrence probabilities of each pathset. Due to the statistical
independence assumed, the probability of occurrence of a pathset is simply
the product of the operation probabilities of the edges in the pathset and
the failure probabilities of the edges not present in the pathset. Complete
state enumeration requires the generation of all $2^{\left\vert E\right\vert
}$ states of ${\mathbb{G}},$ implying that the running time of this
algorithm would exponentially depend on the number of links in the network.

A substantial amount of effort has been put into finding more efficient
algorithms for exact calculation of network reliability problems (see \cite%
{colbournCombinatorics}). However, efficient exact solutions seem unlikely:

\begin{theorem}
\label{thm:npHard} The problems $\Pi _{\text{two terminal}}$ and $\Pi _{ 
\text{rel poly}}$ are \textbf{NP}-hard.
\end{theorem}

For a proof of this result, see, for instance, Theorem 1 in \cite{Ball80}.
These problems belong to the class \textbf{\#P}-complete\ \cite%
{Buzacott,BallPHD,Rosenthal,Ball79,Agrawal,colbournCombinatorics}. \textbf{%
\#P}\ is the set of the counting problems associated with the decision
problems in the set \textbf{NP}. Thus, while a decision problem might ask
whether something exists (e.g., an assignment of truth values to a set of
variables which satisfies a given formula), the corresponding enumeration
problem asks how many of these there are. Solving the enumeration problem
solves the corresponding decision problem since knowing whether the number
of these things is positive tells us whether one exists. In particular, the
counting version of any problem is always at least as hard as the
corresponding existence problem. In analogy to \textbf{NP}-completeness, a
problem is \textbf{\#P}-complete if and only if it is in \textbf{\#P}, and
every problem in \textbf{\#P} can be reduced to it by a polynomial-time
counting reduction (see \cite{gareyComputers} for more details).

\begin{corollary}
\label{thm:sirIsNpHard} The problem $\Pi _{\text{epidemic}}$ is \textbf{NP}%
-hard.
\end{corollary}

To see this, notice that every instance of $\Pi _{\text{two terminal}}$ is
an instance of $\Pi _{\text{epidemic}}$, namely, an instance in which $P$
consists of a single vertex.

More generally, despite dedicated efforts, no algorithm of polynomial
running time has been found that allows for the exact calculation of the
probability $\Pr (Op(G))$ of a given set of pathsets $Op(G),$ unless very
specific assumptions are made on the topology of the underlying
probabilistic network (\cite{colbournCombinatorics,NetRel_Harms}). We
consider it an open question as to which (if any) of these more general
network reliability problems (defined through the choice of a suitable
stochastic binary system $Op(G)$) correspond to epidemiological problems.


\section{\textbf{NP}-hardness of extended\ problems in epidemiology}

\label{sec:problemsInEpidemiology}

\subsection*{Epidemic on networks with time-varying transmission
probabilities}

As we have seen in the previous section, the seemingly simple problem of
finding an individual's chances of infection is \textbf{NP}-hard. This is
even so in the case where the set of initial infectives is a single
individual.

We can generalize $\Pi _{\text{epidemic}}$ by allowing transmission
probabilities to vary over time. We have seen that the length of any
epidemic is at most the length of the longest self-avoiding path in $G$.
Consequently, time-varying transmission probabilities can be encoded as 
\begin{equation*}
{\mathop{\text{Pr}}}:E\times \{1,\dots ,\left\vert E\right\vert
\}\rightarrow \lbrack 0,1].
\end{equation*}%
In this case, percolation methods no longer apply. However, every instance
of $\Pi _{\text{epidemic}}$ can be mapped into an instance of this extended
problem. Thus, the time-varying version of this problem is \textbf{NP}-hard.

\subsection*{Epidemic on networks with disease latency}

One might also generalize $\Pi _{\text{epidemic}}$ to allow patterns of
latency and extended periods of infectivity\footnote{%
For a more general version of this see \cite{Floyd}.}. We will take $%
\mathcal{I}$ to be a sequence of distinct states, $\{I_{1},I_{2},\dots
,I_{N}\}$. We assume that for each stage $I_{i}$ there is an infectivity $%
\mu _{i}$ and a probability of recovery $\rho _{i}$. We take $\rho _{N}=1$.
We now consider a social contact network $\mathcal{G}$ and infectivity pattern $%
\mathcal{I}$. We refer to this as an \emph{$\{S,\mathcal{I},R\}$ network}.
The states of this network are 
\begin{equation*}
\{\phi \mid \phi :V\rightarrow \{S\}\cup \mathcal{I}\cup \{R\}\}.
\end{equation*}
We modify the definition of possible successor states so that the allowable
transitions are from $S$ to $I_{1}$, from $I_{i}$ to $I_{i+1}$ for $%
i=1,\dots ,N-1$ and from $I_{i}$ to $R$ for $i=1,\dots ,N$. If $\phi
(u)=I_{i}$, $u$ transitions to state $R$ with probability $\rho _{i}$ and to
state $I_{i+1}$ with probability $1-\rho _{i}$. If $e=(u,v)\in E$ and $\phi
(u)=I_{i}$, and $\phi (v)=S$, then $u$ infects $v$ with probability ${\ %
\mathop{\text{Pr}}}_{\mathcal{I}}(e,i)=\mu _{i}{\mathop{\text{Pr}}}(e)$. We
assume that $\mathcal{I}$ is non-trivial in the sense that there is $i$ with 
$\mu _{i}\neq 0$ and $\rho _{j}\neq 1$. This ensures that an infected
individual has a positive probability of reaching an infective state. As
before, under the assumption that transmissions and recoveries happen
independently, we can develop an expression for ${\mathop{\text{Pr}}}_{%
\mathcal{I}}(u\mid P)$.

Fix $\mathcal{I}$. An instance of $\Pi_\cali$ is an instance of $\Pi_{\text{
epidemic}}$.

A solution to $\Pi_\cali$ is the value ${\mathop{\text{Pr}}}_\cali(u \mid P)$

\begin{theorem}
\label{thm:calI} \label{thm:latency} Given an non-trivial infectivity
pattern $\mathcal{I}$, $\Pi_\cali$ is \textbf{NP}-hard.
\end{theorem}

\begin{lemma}
Given $\mathcal{G} =(G,{\mathop{\text{Pr}}})$ and $\mathcal{I}$, there is $%
\mathcal{G}^{\prime}= (G,{\mathop{\text{Pr}}}^{\prime})$ so that for each $%
P\subset V$ and $u\notin P$, ${\mathop{\text{Pr}}}_\cali(u\mid P) = {\ %
\mathop{\text{Pr}}}^{\prime}(u \mid P)$.
\end{lemma}

\begin{proof}
Consider an edge $e=(u,v)$. Suppose that $\phi _{1}(u)=I_{i}$ and $\phi
_{1}(v)=S$. What are the chances that $v$ remains uninfected by $u$? (We
assume for the moment that $v$ is not infected by some other neighbour
during the next $N$ steps.) We take $\mu ={\mathop{\text{Pr}}}(e)$. Let us
denote by $\nu _{i}$ the probability that $u$ remains infected for $i$
steps, but not $i+1$ steps. We then have 
\begin{equation*}
\nu _{i}=\rho _{i}\prod_{j=1}^{i}(1-\rho _{j}).
\end{equation*}%
The probability that $v$ remains uninfected by $u$ is 
\begin{equation*}
\tau _{\mathcal{I}}(\mu )=\sum_{i=1}^{N}\nu _{j}\prod_{j=1}^{i}(1-\mu
_{j}\mu )
\end{equation*}

We now define $\mathcal{G}^{\prime}= (G,{\mathop{\text{Pr}}}^{\prime})$ by
taking 
\begin{equation*}
{\mathop{\text{Pr}}}^{\prime}(e) = 1 - \tau_\cali({\mathop{\text{Pr}}}(e)).
\end{equation*}
This does what is required.
\end{proof}

\begin{proof}[Proof of Theorem~\protect\ref{thm:latency}]
We will show that $\Pi_{\text{rel poly}}$ is polynomially reducible to $%
\Pi_\cali$.

Fix $\mathcal{I}$ to be a non-trivial pattern of infectivity. Suppose we are
given an instance $\pi$ of $\Pi_{\text{rel poly}}$. This consists of a graph 
$G$ and source and target vertices $u$ and $v$. Suppose also that we have a
polynomial time algorithm for solving $\Pi_\cali$. We choose $N+1$ arbitrary
probabilities $p_0,\dots,p_{N+1}$. These give us $N+1$ instances of $%
\Pi_\cali$ by taking $\mathcal{G}_i = (G,{\mathop{\text{Pr}}}_i)$, where ${\ %
\mathop{\text{Pr}}}_i$ takes the constant value $p_i$. By the previous
lemma, solving these $N+1$ instances of $\Pi_\cali$ solves $N+1$ distinct
instances of $\Pi_{\text{epidemic}}$ which consist of the graph $G$ and
differing constant functions ${\mathop{\text{Pr}}}_i^{\prime}$. These $N+1$
values give us $N+1$ independent linear equations whose unknowns are the
coefficients of the reliability polynomial. Solving for these is a
polynomial time problem.
\end{proof}

\subsection*{Expected number of total infections}

One might hope that while computing an individual's probability of infection
is \textbf{NP}-hard, there might be a way to compute the expected number of
infections. This, too, is \textbf{NP}-hard. Let us formalize this.

An instance $\pi$ of $\Pi_{\text{expected}}$ is

\begin{itemize}
\item A graph $G = (V,E)$.

\item A labelling ${\mathop{\text{Pr}}}:E\rightarrow \lbrack 0,1]\cap 
\mathbb{Q}
$.

\item An initial infective set $P \subset V$.
\end{itemize}

A solution to $\pi $ is the expected number of infections, 
\begin{equation*}
\sum_{u\in V}{\mathop{\text{Pr}}}(u\mid P).
\end{equation*}%
The following theorem was proved in \cite{Laumanns_Zenklusen_2009}. For the
sake of completeness, we provide a proof here.

\begin{theorem}
\label{thm:expected} $\Pi_{\text{expected}}$ is \textbf{NP}-hard.
\end{theorem}

\begin{proof}
We will show that $\Pi _{\text{epidemic}}$ can be polynomially reduced to $%
\Pi _{\text{expected}}$. Suppose we are given an instance $\pi $ of $\Pi _{ 
\text{epidemic}}$. Let $\tilde{\pi}$ be the instance of $\Pi _{\text{%
epidemic }}$ which is formed from $\pi $ by appending a single edge from $u$
to $v\notin V$ and assigning ${\mathop{\text{Pr}}}(u,v)=1$. It is clear that
the expected number of infections in $\tilde{\pi}$ differs from the number
of expected infections in $\pi $ by exactly ${\mathop{\text{Pr}}}(u\mid P)$.
Thus, if we had a polynomial time algorithm for finding the expected number
of infections, we could find the probability of any individual becoming
infected.
\end{proof}

The fact that $\Pi_{\text{rel poly}}$ is \textbf{NP}-hard\ suggests that the
difficulty lies not in the probabilities ${\mathop{\text{Pr}}}$ but in the
topology of $G$. One problem which we have not addressed here is the
question of calculating the probability of infection in an $\{S,I,R\}$
network where $G = G_t$ changes over time due to stochastic births and
deaths. It seems likely that this will also provide a source of \textbf{NP}
-hard\ problems. However, this requires a reformulation of the underlying
problem.


\section{Discussion and conclusions}

\label{sec:discussion}

It has been the purpose of this paper to draw the attention of network
epidemiologists to results in communications network reliability which shed
light on questions regarding the computational aspects of epidemiology of $%
\{S,I,R\}$ networks.

Theorem~\ref{thm:sirIsNpHard} and Theorem~\ref{thm:expected} tell us that
generally, in the absence of a major break-through in computer science we
cannot expect to be able to compute \textit{exact} probabilities of
infection or expected number of infection in large social contact networks. As \cite%
{gareyComputers} points out, problems do not go away simply because we have
deemed them \textbf{NP}-hard.

Since the network engineers have been here before us, it is tempting to ask
whether their solutions will work for epidemiologists. While we consider the
case open, the prospects seem mixed. Network engineers are often in the
position of being able to choose the class of networks under consideration.
As opposed to scale-free \cite{eubankNetworkBased,dynamicsOnScaleFree} and
small-world network structures \cite%
{strogatz,newmanScaling,eubankNetworkBased}, which frequently arise from a
self-organization process during the spontaneous growth of a network,
engineered or purposefully designed networks show rather different
structures. Some of the classes that allow efficient calculations (exact or
approximate) include trees, full graphs, series-parallel graphs \cite%
{colbournCombinatorics}, and channel graphs \cite{NetRel_Harms}.
Unfortunately, these classes of networks seem unrealistic as models of
social contact networks.

Network engineers have turned to Monte Carlo simulation for the calculation
of estimates of network reliability. We would like to give pointers into
their literature \cite%
{MR0438625,Gertsbakh,MR874297,MR1691978,475978,MR1959383,MR2136660,Laumanns_Zenklusen_2009}%
. This approach has received increased attention in the last decade due to
the power of modern computers and computing clusters. While Monte Carlo
simulation only calculates an unbiased point estimator for reliability
probabilities, increasing the number of simulated samples causes these
estimates to converge to the actual value.

The fact that efficient and precise algorithms for computing infection
probabilities are out of reach (see Theorems \ref{thm:sirIsNpHard}, \ref%
{thm:calI} and \ref{thm:expected}) has real-world consequences. Designing a
response to an emerging epidemic can depend on determining the kind of
epidemiological probabilities we have been discussing \cite%
{optimizingInfectiousDiseaseInterventions}. The effectiveness of
interventions during an emerging epidemic often crucially depends on timely
implementation. Our results and those of \cite{wangModeling} and \cite%
{eubankStructural} place an emphasis on the search for efficient and quick
methods that give good approximations when applied to real-world social
networks.



\begin{thebibliography}{99}
\bibitem{KM} Kermack WO, McKendrick AG. \newblock A contribution to the
mathematical theory of epidemics. \newblock Proc Roy Soc Lond A.
1927;115:700--721.

\bibitem{Bansal22102007} Bansal, Shweta, Grenfell, T B, Meyers, Ancel L. %
\newblock When individual behaviour matters: homogeneous and network models
in epidemiology. \newblock Journal of The Royal Society Interface.
2007;4(16):879--891. \newblock Available from: %
\url{http://rsif.royalsocietypublishing.org/content/4/16/879.abstract}.

\bibitem{N1} Newman MEJ. \newblock Spread of epidemic disease on networks. %
\newblock Phys Rev E. 2002 Jul;66(1):016128.

\bibitem{Eubank2004} Eubank S, Guclu H, Kumar VSA, Marathe MV, Srinivasan A,
Toroczkai Z, et~al. \newblock Modelling disease outbreaks in realistic urban
social networks. \newblock Nature. 2004;429(6988):180--184. \newblock %
Available from: \url{http://www.ncbi.nlm.nih.gov/pubmed/15141212}.

\bibitem{Cauchemez15022011} Cauchemez S, Bhattarai A, Marchbanks TL, Fagan
RP, Ostroff S, Ferguson NM, et~al. \newblock Role of social networks in
shaping disease transmission during a community outbreak of 2009 H1N1
pandemic influenza. \newblock Proceedings of the National Academy of
Sciences. 2011;108(7):2825--2830. \newblock Available from: %
\url{http://www.pnas.org/content/108/7/2825.abstract}.

\bibitem{Keeling22092005} Keeling MJ, Eames KTD. \newblock Networks and
epidemic models. \newblock Journal of The Royal Society Interface.
2005;2(4):295--307. \newblock Available from: %
\url{http://rsif.royalsocietypublishing.org/content/2/4/295.abstract}.

\bibitem{Newman03thestructure} Newman MEJ. \newblock The Structure and
Function of Complex Networks. \newblock SIAM Review. 2003;45:167--256.

\bibitem{eubankNetworkBased} Eubank S. \newblock Network based models of
infectious disease spread. \newblock Jpn J Infect Dis. 2005;58(6):S9--13.

\bibitem{Keeling20051} Keeling M. \newblock The implications of network
structure for epidemic dynamics. \newblock Theoretical Population Biology.
2005;67(1):1 -- 8. \newblock Available from: %
\url{http://www.sciencedirect.com/science/article/B6WXD-4F4NYDN-1/2/58290fd5aa19b724c74dda6a1aa296d2}%
.

\bibitem{PhysRevE.79.036113} Allard A, No\"el PA, Dub\'e LJ, Pourbohloul B. %
\newblock Heterogeneous bond percolation on multitype networks with an
application to epidemic dynamics. \newblock Phys Rev E. 2009
Mar;79(3):036113.

\bibitem{PhysRevE.80.020901} Miller JC. \newblock Percolation and epidemics
in random clustered networks. \newblock Phys Rev E. 2009 Aug;80(2):020901.

\bibitem{Eur.Phys.J.B_26} {Y Moreno}, {R Pastor-Satorras}, {A Vespignani}. %
\newblock Epidemic outbreaks in complex heterogeneous networks. \newblock %
Eur Phys J B. 2002;26(4):521--529. \newblock Available from: %
\url{http://dx.doi.org/10.1140/epjb/e20020122}.

\bibitem{Meyers2006400} Meyers LA, Newman MEJ, Pourbohloul B. \newblock %
Predicting epidemics on directed contact networks. \newblock Journal of
Theoretical Biology. 2006;240(3):400 -- 418. \newblock Available from: %
\url{http://www.sciencedirect.com/science/article/pii/S0022519305004418}.

\bibitem{SIRinRandomNetworks} Volz E. \newblock SIR dynamics in random
networks with heterogeneous connectivity. \newblock Journal of Mathematical
Biology. 2008;56:293--310. \newblock 10.1007/s00285-007-0116-4. \newblock %
Available from: \url{http://dx.doi.org/10.1007/s00285-007-0116-4}.

\bibitem{Volz07122007} Volz E, Meyers LA. \newblock Susceptible\^{a}\euro
``infected\^{a}\euro ``recovered epidemics in dynamic contact networks. %
\newblock Proceedings of the Royal Society B: Biological Sciences.
2007;274(1628):2925--2934. \newblock Available from: %
\url{http://rspb.royalsocietypublishing.org/content/274/1628/2925.abstract}.

\bibitem{PhysRevLett.86.3200} Pastor-Satorras R, Vespignani A. \newblock %
Epidemic Spreading in Scale-Free Networks. \newblock Phys Rev Lett. 2001
Apr;86(14):3200--3203.

\bibitem{House06012011} House T, Keeling MJ. \newblock Insights from
unifying modern approximations to infections on networks. \newblock Journal
of The Royal Society Interface. 2011;8(54):67--73. \newblock Available from: %
\url{http://rsif.royalsocietypublishing.org/content/8/54/67.abstract}.

\bibitem{PhysRevLett.90.028701} Bogu\~n\'a M, Pastor-Satorras R, Vespignani
A. \newblock Absence of Epidemic Threshold in Scale-Free Networks with
Degree Correlations. \newblock Phys Rev Lett. 2003 Jan;90(2):028701.

\bibitem{stroud2007} Stroud P, Del~Valle S, Sydoriak S, Riese J, Mniszewski
S. \newblock Spatial Dynamics of Pandemic Influenza in a Massive Artificial
Society. \newblock Journal of Artificial Societies and Social Simulation.
2007;10(4):9. \newblock Available from: %
\url{http://jasss.soc.surrey.ac.uk/10/4/9.html}.

\bibitem{Mniszewski:2008:ESM:1400549.1400636} Mniszewski SM, Del~Valle SY,
Stroud PD, Riese JM, Sydoriak SJ. \newblock EpiSimS simulation of a
multi-component strategy for pandemic influenza. \newblock In: Proceedings
of the 2008 Spring simulation multiconference. SpringSim '08. San Diego, CA,
USA: Society for Computer Simulation International; 2008. p. 556--563. %
\newblock Available from: %
\url{http://portal.acm.org/citation.cfm?id=1400549.1400636}.

\bibitem{EpiSimSLA} SY Del Valle, D Kubicek, SM Mniszewski, JM Riese, PR
  Romero, JP Smith, et~al. \newblock EpiSimS Los Angeles Case Study. \newblock
  Los Alamos National Laboratory; 2006. LAUR-06-0666.

\bibitem{Mniszewski:2008:PSA:1394977.1394987} Mniszewski SM, Del~Valle SY,
Stroud PD, Riese JM, Sydoriak SJ. \newblock Pandemic simulation of
antivirals + school closures: buying time until strain-specific vaccine is
available. \newblock Comput Math Organ Theory. 2008 September;14:209--221. %
\newblock Available from: %
\url{http://portal.acm.org/citation.cfm?id=1394977.1394987}.

\bibitem{Wallinga15_2006} Wallinga J, Teunis P, Kretzschmar M. \newblock %
Using Data on Social Contacts to Estimate Age-specific Transmission
Parameters for Respiratory-spread Infectious Agents. \newblock American
Journal of Epidemiology. 2006;164(10):936--944. \newblock Available from: %
\url{http://aje.oxfordjournals.org/content/164/10/936.abstract}.

\bibitem{wangModeling} Wang N. \newblock Modeling and analysis of massive
social networks [Ph.D Thesis]. \newblock University of Maryland. College
Park, MD, USA; 2005. \newblock Adviser: Srinivasan, Aravind.

\bibitem{eubankStructural} Eubank S, Kumar VSA, Marathe MV, Srinivasan A,
Wang N. \newblock Structural and algorithmic aspects of massive social
networks. \newblock In: SODA '04: Proceedings of the fifteenth annual
ACM-SIAM symposium on Discrete algorithms. Philadelphia, PA, USA: Society
for Industrial and Applied Mathematics; 2004. p. 718--727.

\bibitem{aspnesSODA} Aspnes J, Chang K, Yampolskiy A. \newblock Inoculation
strategies for victims of viruses and the sum-of-squares partition problem. %
\newblock In: SODA '05: Proceedings of the sixteenth annual ACM-SIAM
symposium on Discrete algorithms. Philadelphia, PA, USA: Society for
Industrial and Applied Mathematics; 2005. p. 43--52.

\bibitem{hayrapetyanUnbalancedGraphCuts} Hayrapetyan A, Kempe D, P\'{a}l M,
Svitkina Z. \newblock Unbalanced Graph Cuts. \newblock In: Algorithms - ESA
2005. Springer-Verlag; 2005. p. 191--202. \newblock Available from: %
\url{{http://dx.doi.org/10.1007/11561071_19}}.

\bibitem{Laumanns_Zenklusen_2009} Laumanns M, Zenklusen R. \newblock %
Computational complexity of impact size estimation for spreading processes
on networks. \newblock European Physical Journal B. 2009;71(4):481--487. %
\newblock Available from: %
\url{http://www.springerlink.com/index/10.1140/epjb/e2009-00344-7}.

\bibitem{gareyComputers} Garey MR, Johnson DS. \newblock Computers and
Intractability; A Guide to the Theory of NP-Completeness. \newblock New
York, NY, USA: W. H. Freeman \& Co.; 1990.

\bibitem{HopcroftUllman} Hopcroft JE, Ullman JD. \newblock Introduction To
Automata Theory, Languages, And Computation. \newblock Boston, MA, USA:
Addison-Wesley Longman Publishing Co., Inc.; 1990.

\bibitem{origami} Bern M, Hayes B. \newblock The complexity of flat origami. %
\newblock In: Proceedings of the 7th Annual ACM-SIAM Symposium on Discrete
Mathematics; 1996. p. 175--183.

\bibitem{Floyd} Floyd W, Kay L, Shapiro M. \newblock Some elementary
properties of SIR networks or, can {I} get sick because you got vaccinated? %
\newblock Bull Math Biol. 2008 Apr;70(3):713--727. \newblock Available from: %
\url{http://www.ncbi.nlm.nih.gov/pubmed/18060461}.

\bibitem{colbournCombinatorics} Colbourn CJ. \newblock The Combinatorics of
Network Reliability. \newblock New York, NY, USA: Oxford University Press,
Inc.; 1987.

\bibitem{G} Grassberger P. \newblock Critical behavior of the general
epidemic process and dynamical percolation. \newblock Math Biosci.
1983;63:205--213.

\bibitem{S} Sander LM, Warren CP, Sokolov IM, Simon C, Koopman J. \newblock %
Percolation on heterogeneous networks as a model for epidemics. \newblock %
Math Biosci. 2002;180:293--205.

\bibitem{secondLook} Kenah E, Robbins JM. \newblock Second look at the
spread of epidemics on networks. \newblock Phys Rev E. 2007;76:036113. %
\newblock Available from: \url{doi:10.1103/PhysRevE.76.036113}.

\bibitem{M} Meyers LA. \newblock Contact network epidemiology: bond
percolation applied to infectious disease prediction and control. \newblock %
Bull Amer Math Soc. 2007;44:63--86.

\bibitem{Ball80} Ball MO. \newblock Complexity of network reliability
computations. \newblock Networks. 1980;10:153--165.

\bibitem{Buzacott} Buzacott JA. \newblock A Recursive Algorithm for Finding
the Probability that a Graph is Disconnected. \newblock Department of
Industrial Engineering, University of Toronto; 1976. Working Paper No. 76-0
16.

\bibitem{BallPHD} Ball MO. \newblock Network Reliability Analysis:
Algorithms and Complexity [Ph.D. Thesis]. \newblock Cornell University; 1977.

\bibitem{Rosenthal} Rosenthal A. \newblock Computing the Reliability of
Complex Networks. \newblock SIAM J Appl Math. 1977;32:384--393.

\bibitem{Ball79} Ball MO. \newblock Computing network reliability. \newblock %
Oper Res. 1979;27:823--838.

\bibitem{Agrawal} Agrawal A, Barlow RE. \newblock A survey of network
reliability and domination theory. \newblock Oper Res. 1984;32:478--492.

\bibitem{NetRel_Harms} Harms DD, Kraetzl M, Colbourn CJ, Devitt JS. %
\newblock Network Reliability: Experiments with a Symbolic Algebra
Environment. \newblock Boca Raton, FL, USA: CRC Press, Inc.; 1995.

\bibitem{dynamicsOnScaleFree} May RM, Lloyd AL. \newblock Infection dynamics
on scale-free networks. \newblock Phys Rev E. 2001 Nov;64(6):066112.

\bibitem{strogatz} Watts DJ, Strogatz SH. \newblock Collective dynamics of
'small-world' networds. \newblock Nature. 1998;393(4):440--443.

\bibitem{newmanScaling} Newman MEJ, Watts DJ. \newblock Scaling and
percolation in the small-world network model. \newblock Phys Rev E.
1999;60:7332--7342.

\bibitem{MR0438625} Barlow RE, Proschan F. \newblock Statistical theory of
reliability and life testing. \newblock Rinehart and Winston, Inc., New
York: Holt; 1975. \newblock Probability models, International Series in
Decision Processes, Series in Quantitative Methods for Decision Making.

\bibitem{Gertsbakh} Gertsbakh IB, Shpungin Y. \newblock Models of Network
Reliability: Analysis, Combinatorics and {M}onte {C}arlo. \newblock Boca
Raton, FL, USA: CRC Press, Inc.; 2010.

\bibitem{MR874297} Fishman GS. \newblock A {M}onte {C}arlo sampling plan for
estimating network reliability. \newblock Oper Res. 1986;34(4):581--594. %
\newblock Available from: \url{http://dx.doi.org/10.1287/opre.34.4.581}.

\bibitem{MR1691978} Lomonosov M, Shpungin Y. \newblock Combinatorics of
reliability {M}onte {C}arlo. \newblock Random Structures Algorithms.
1999;14(4):329--343. \newblock Available from: %
\url{http://dx.doi.org/10.1002/(SICI)1098-2418(199907)14:4<329::AID-RSA3>3.0.CO;2-X}%
.

\bibitem{475978} Cancela H, El~Khadiri M. \newblock A recursive
variance-reduction algorithm for estimating communication-network
reliability. \newblock Reliability, IEEE Transactions on. 1995 dec;44(4):595
--602.

\bibitem{MR1959383} Hui KP, Bean N, Kraetzl M, Kroese D. \newblock The tree
cut and merge algorithm for estimation of network reliability. \newblock %
Probab Engrg Inform Sci. 2003;17(1):23--45. \newblock Available from: %
\url{http://dx.doi.org/10.1017/S0269964803171021}.

\bibitem{MR2136660} Hui KP, Bean N, Kraetzl M, Kroese DP. \newblock The
cross-entropy method for network reliability estimation. \newblock Ann Oper
Res. 2005;134:101--118. \newblock Available from: %
\url{http://dx.doi.org/10.1007/s10479-005-5726-x}.

\bibitem{optimizingInfectiousDiseaseInterventions} Wallinga J, van Boven M,
Lipsitch M. \newblock Optimizing infectious disease interventions during an
emerging epidemic. \newblock PNAS. 2010;107(2):923--928.
\end{thebibliography}


\end{document}